\documentclass[12pt,a4paper]{article}
\usepackage{amsmath}
\usepackage{amssymb}
\usepackage{amsthm}
\usepackage{authblk}
\usepackage[font=footnotesize]{caption}
\usepackage{cite}
\usepackage{dsfont}
\usepackage{enumerate}
\usepackage[top=2.5cm, bottom=3cm, left=2.8cm, right=2.8cm]{geometry}
\usepackage{graphicx}
\usepackage[colorlinks=true,linkcolor=blue]{hyperref}
\usepackage[font=footnotesize]{subcaption}
\usepackage[capitalize]{cleveref}

\newcommand{\C}{\mathbb{C}}
\newcommand{\R}{\mathbb{R}}
\newcommand{\1}{\mathds{1}}
\newcommand{\cA}{\mathfrak{A}}
\newcommand{\cC}{\mathfrak{C}}
\newcommand{\cD}{\mathcal{D}}
\newcommand{\cH}{\mathcal{H}}
\newcommand{\cM}{\mathcal{M}}
\newcommand{\cO}{\mathcal{O}}
\newcommand{\x}{\mathbf{x}}
\newcommand{\y}{\mathbf{y}}
\newcommand{\z}{\mathbf{z}}
\newcommand{\g}{\mathbf{g}}
\newcommand{\bk}{\mathbf{k}}
\newcommand{\SU}{\mathrm{SU}}
\newcommand{\ii}{\mathrm{i}}

\AddToHook{cmd/appendix/before}{
  \setcounter{equation}{0}
  
}

\newtheoremstyle{theoremstyle}{3pt}{3pt}{}{}{\bfseries}{:}{0.5em}{}

\theoremstyle{theoremstyle}
\newtheorem{condition}{Condition}
\crefname{condition}{Condition}{Conditions}

\theoremstyle{theoremstyle}
\newtheorem{proposition}{Proposition}
\crefname{proposition}{Proposition}{Propositions}

\theoremstyle{theoremstyle}
\newtheorem{theorem}{Theorem}

\theoremstyle{theoremstyle}
\newtheorem{lemma}{Lemma}

\def\be#1\ee{\begin{equation}#1\end{equation}}
\def\bea#1\eea{\begin{align}#1\end{align}}

\title {\Large \bf Necessary and sufficient conditions for correctness of complex Langevin} 
\author[1]{Michael Mandl\thanks{\href{mailto:michael.mandl@uni-graz.at}{michael.mandl@uni-graz.at}}}
\author[2]{Erhard Seiler\thanks{\href{mailto:ehs@mpp.mpg.de}{ehs@mpp.mpg.de}}}
\author[1]{Dénes Sexty\thanks{\href{mailto:denes.sexty@uni-graz.at}{denes.sexty@uni-graz.at}}} 
\affil[1]{Institute of Physics, NAWI Graz, University of Graz, Universit{\"a}tsplatz 5, 8010 Graz, Austria}
\affil[2]{Max-Planck-Institut f{\"u}r Physik (Werner-Heisenberg-Institut), Boltzmannstra{\ss}e 8,\newline 85748 Garching bei M{\"u}nchen, Germany}
\date{
\vspace{-0.7cm}\normalsize\today}

\begin{document}

\maketitle

We derive a family of correctness conditions for complex Langevin simulations. In particular, we show that if in a given theory the expectation
values of all observables within a particular space satisfy the theory's Schwinger--Dyson equations as well as certain bounds, then these expectation values are necessarily correct. In fact, these findings are not only valid in the context of complex Langevin simulations, but they also hold for general probability densities on complex manifolds, given an initial complex density on a real manifold. We stress that, while the proposed conditions are necessary and sufficient in a mathematical sense, their practical use is not to prove the correctness of obtained simulation results. Rather, they are mainly useful for detecting incorrect convergence. In particular, we test these criteria in a few simple one- and two-dimensional toy models and find that they are indeed capable of ruling out incorrect results without the need of exact solutions. 

\section{Introduction}\label{sec:introduction}
    The numerical sign problem remains among the most challenging problems in many-particle 
physics to this day. It prevents the numerical evaluation of integrals over highly 
oscillating functions with a large number of degrees of freedom with standard techniques. Such 
integrals arise, for instance, in the nonperturbative study of strongly interacting quantum
field theories with a nonzero density of fermions by means of first principle lattice 
simulations. Moreover, they are also prevalent in lattice approaches to the study of 
real-time quantum field theories, as well as in various other theories of interest in 
physics.

The complex Langevin method \cite{Par83,Kla83} can be used -- in principle -- in situations
where conventional lattice approaches cannot be applied due to the sign problem, thus 
providing a potential solution. For a theory whose path integral consists of the integration  of an oscillatory integrand $\rho$ over a real manifold $\cM_r$, the complex Langevin approach makes use of the stochastic evolution in a complexified
field space $\cM_c$. Effectively, it replaces $\rho$ by a real and nonnegative distribution $P$ on $\cM_c$, which thus allows 
for a probabilistic interpretation similar to the one used in usual Monte-Carlo approaches. 
The complex Langevin evolution is considered successful if and only if $P$ is equivalent to
$\rho$ in the sense that both produce the same expectation values on an appropriate space of 
holomorphic observables. The existence of such a probability density $P$ for a given
$\rho$ has been established quite generally \cite{Sal97,Wei02,Sal07}. What is less clear, 
however, is whether complex Langevin is capable of producing such $P$ for generic $\rho$. 
In fact, complex Langevin simulations may suffer from the infamous wrong convergence 
problem, meaning that they might fail to reproduce the correct expectation values, at least 
without additional modifications. 

Establishing the correctness of complex Langevin simulation results is thus of utmost
importance in order to ensure the reliability of the method. While there is a plethora
of potential correctness criteria to be found in the literature, as we shall discuss, none
of them appear to be entirely satisfactory and free of loopholes. In this work, we attempt 
to fill this gap by deriving a set of conditions that are both necessary and sufficient to 
ensure that a probability density $P$ on $\cM_c$ produces the same expectation values as a 
given complex density $\rho$ on $\cM_r$. The correctness criterion proposed here can thus be 
applied to complex Langevin results, but it is, in fact, more general than that. Indeed, it can be used in any situation where a complex density on a real manifold $\cM_r$ is traded for a probability density on the complexification of $\cM_r$, see, e.g., \cite{Wei02,Sal07,Wos15,SW17}.

This work is structured at follows: After introducing a number of useful definitions and 
terminology in \cref{sec:generalities}, we derive the new correctness condition in 
\cref{sec:criteria}. In \cref{sec:cle}, we introduce the complex Langevin approach, which we then use to test the correctness condition. The results of this are presented in \cref{sec:numerics}. Finally, in 
\cref{sec:realistic}, we provide a few ideas on how to employ the criterion in more 
realistic models than the ones considered here, before concluding in \cref{sec:summary}. 
The appendices contain some auxiliary results and derivations.
\section{Generalities}\label{sec:generalities}
    \subsection{Manifolds and densities}
We consider models defined on a real manifold $\cM_r$, for instance the configuration space of
a lattice gauge theory. The manifold $\cM_r$ generally is of the form
\be\label{eq:M_r}
\cM_r=\R^n\times G^m\;,
\ee
where $G$ is a compact Lie group such as $\SU(N)$. The case $n=0$ is interpreted as the 
absence of the factor $\R^n$, while $m=0$ is is interpreted as the absence 
of the factor $G^m$. On $\cM_r$, we consider the measure 
\be
d\mu\equiv d^nx\times d^mg\;,
\ee
i.e., the product of the Lebesgue measure $d^nx$ on $\R^n$ with the Haar measure $d^mg$ on $G^m$. $\cM_r$ has a complexification to a complex manifold $\cM_c$,
\be
\cM_c=\C^n\times G_c^m\;,
\ee
where $G_c$ is the complexification of $G$. We 
denote the arguments of functions on $\cM_r$ by $(\x,\g)$ and when complexified by 
$(\z,\g_c)$. Occasionally, for notational brevity we will omit the arguments of functions in this work when they are evident from the context.

The models we consider are defined by a complex density $\rho(\x,\g)$, which has a holomorphic
continuation to $\cM_c$. In many cases, $\rho$ is given in terms of an action $S(\x,\g)$ via
\be
\rho(\x,\g)\propto\exp(-S(\x,\g))\;.
\ee
We assume $\rho$ to be normalized, such that 
\be
\int_{\cM_r} \rho\,d\mu=1\;.
\ee
If $\cM_r$ is noncompact, i.e., $n>0$, we furthermore assume that 
\be
\rho_m(\x)\equiv \sup_{\g\in G^m} |\rho(\x,\g)|
\label{rho_m}
\ee
goes to zero at least as fast as a Gaussian for $|\x|\to\infty$, i.e.,
\be
\rho_m(\x)\le A \exp(-a\vert\x\vert^2)\;,
\label{gauss}
\ee
with some positive constants $A$ and $a$. This condition could likely be weakened, but it helps to keep the discussion as simple as possible.
In any case, the condition is always satisfied in the usual lattice models.

Another requirement that is met for most conventional models, which we also impose here, is that the first partial derivatives of $\log\rho$ with respect to both the coordinates and the group degrees of freedom (see below) grow at most like a polynomial in $\x$ for all $\g$. Finally, we assume that $\rho$ does not vanish on $\cM_r$:
\be
\rho(\x,\g)\neq 0\quad \forall\,(\x,\g)\in\cM_r\;. 
\ee
This is not really a restriction because by Cauchy's theorem
we can always deform the integration manifold $\cM_r$ away from any zeros, for instance by shifting by some vector with small purely imaginary components, without changing the expectation values of observables.
  
\subsection{Observables and Schwinger--Dyson operators}
Let us consider the (commutative) algebra $\cA$ of observables generated by the coordinates 
of $\R^n$ and the matrix elements\footnote{Notice that the matrix elements of $\g\in G$ are functions of $\g$, such that the observables in $\cA$ are really functions on $\cM_r$.} of the finite-dimensional representations of $G^m$ (note that in the case of $G=\SU(N)$
we only require the matrix elements of the fundamental representation as generators of the algebra).
All the elements $f\in\cA$ have analytic continuations to $\cM_c$, which 
we denote by the same letter $f$. The constant function with value $1$ is the unit 
element of $\cA$, denoted by $\1$. It turns out, however, that we actually require a somewhat larger space of observables for our purposes, as will be explained 
below.

It is well known that the correct expectation values of observables have to satisfy the 
so-called Schwinger--Dyson equations (SDEs), which are infinitely many identities 
characterizing $\rho$. The Schwinger--Dyson (SD) operators are differential operators 
acting on the observables, but -- in general -- they do not map $\cA$ into itself. To 
describe the SDEs in this general setting, we first need to establish a few definitions.

Since the functions $f\in\cA$ depend on two kinds of variables, $\x$ and $\g$, we also require
two kinds of derivatives. The first kind are simply the partial derivatives $\partial_i$
with respect to the coordinates $x_i$ of $\R^n$ or, after complexification, the coordinates
$z_i$ of $\C^n$. The second ones are derivatives with respect to the group variables. They 
are determined by choosing a basis
\be
\{X_j\,|\,j=1,\ldots, d \}\,,\quad d=\dim(G)
\ee
of the Lie algebra $\mathfrak g$ of $G$ and, correspondingly, a basis of the Lie algebra
$\mathfrak g^m$ as
\be
\{X_{j,k}\,|\,j=1,\ldots, d;\, k=1, \ldots, m \}\;.
\ee
With this, we define the derivations (usually referred to as left-derivatives)
\be
D_{j,k}f(\x,\g)\equiv \lim_{\epsilon\to 0} \frac{1}{\epsilon}
\big(f(\x,\exp(\epsilon X_{j,k})\g)-f(\x,\g)\big)
\ee
and extend their action to the analytic continuation of the functions  $f\in\cA$.
Corresponding to the two kinds of partial derivatives, there are two kinds of SD operators:
\be\label{eq:A_i}
A_i=\partial_{i} + \frac{\partial_{i}\rho}{\rho}
\ee
and
\be\label{eq:A_jk}
A_{j,k}=D_{j,k} + \frac{D_{j,k}\rho}{\rho}\;.
\ee
{\bf Remark:} The SD operators map $\cA$ into itself only in the case that the partial derivatives of $\log\rho$ are polynomials. More precisely, if 
$\rho$ vanishes somewhere on $\cM_c$, the action of the SD 
operators on $\cA$ will generically produce poles at the zeros of $\rho$. For this reason, as in \cite{SS19},  
we are forced to enlarge $\cA$ to
\be\label{eq:H}
\cH \equiv \cA+ \sum_iA_i \cA+ \sum_{j,k} A_{j,k} \cA\;.
\ee
It is not hard to see that all $f\in\cH$ are integrable with respect to $|\rho|d\mu$. In \cref{sec:density_zeros}, we shall discuss a specific example in which $\cA\neq\cH$.

Now, the ``correct expectation'' of a 
function $f\in \cH$ is defined by
\be\label{eq:rho_expectation}
\langle f \rangle_\rho \equiv \int_{\cM_r} f \rho\, d\mu\;,
\ee 
where
\be
\langle \1\rangle_\rho=1
\ee
due to the normalization of $\rho$. Indeed, $\langle f\rangle_\rho$ is invariant under continuous deformations or shifts of $\cM_r$, owing to the holomorphicity of both $\rho$ and $f$.

\subsection{Correctness of probability measures on \texorpdfstring{$\cM_c$}{Mc}}\label{sec:measure_correctness}
Let $dP$ be a probability measure on $\cM_c$, which may (but does not have to) be the 
equilibrium measure of a complex Langevin process with or without a kernel (see \cref{sec:cle}).

{\bf Definition:} We say that the probability measure $dP$ is {\it correct} if and only if
\be
\int_{\cM_c} f dP = \langle f \rangle_\rho\quad \forall f\in\cH\;.
\label{correct}
\ee
{\bf Remark:} If $\rho$ vanishes somewhere on $\cM_c$, $dP$ can only be correct if it also
vanishes there, because otherwise the singularities in $f\in\cH$ would cause the left-hand side of (\ref{correct}) to diverge.

We factorize $\rho$ as
\be\label{eq:rho_split}
\rho(\x,\g)=w(\x) \rho_r(\x,\g)\;, 
\ee
such that both factors are continuous, $w>0$, and
\be
\Vert f\Vert_{w}\equiv \sup_{\cM_r} |f(\x,\g)|w(\x)<\infty  \quad \forall f\in\cH \;,
\label{norms}
\ee 
as well as 
\be 
C\equiv \int_{\cM_r} |\rho_r(\x,\g)| d\mu<\infty\;. 
\label{integrable} 
\ee 
Note that the factor $w$ is independent of the group variables $\g$.

To be clear, since we assume the partial derivatives of $\log\rho$ to grow at most polynomially, these conditions require that $w(\x)$ decays faster than any power of $\vert\x\vert$, while (\ref{integrable}) must hold at the same time. In fact,
for reasons stated below, we actually require $w(\x)$ to decay at least like a Gaussian, i.e., 
\begin{equation}\label{eq:w_decay}
    w(\x) \leq B\exp(-b\vert\x\vert^2)\;,
\end{equation}
with positive constants $B$ and $b$, similar to \eqref{gauss}. A simple choice 
for $w$ is 
\be 
w(\x)\equiv \rho_m(\x)^\alpha\;,\quad 
\rho_r(\x,\g)\equiv \rho(\x,\g)/w(\x)\;, \quad 0<\alpha<1\;, 
\label{w_choice}
\ee 
where $\rho_m$ was defined in \eqref{rho_m}. Note that this implies
\be 
\sup_{\g\in G^m} |\rho_r(\x,\g)|=w(\x)^{1/\alpha-1}\;. 
\ee 
With this choice, $w(\x)$ indeed decays like a Gaussian due to \eqref{gauss} and is continuous, as is shown in \cref{app:continuity}.

\section{Correctness conditions}\label{sec:criteria}
    With the definitions introduced in the previous section, we are now in the
position to formulate two conditions, the combination of which is both necessary and sufficient for ensuring the correctness of the probability measure $dP$. We proceed by proving the necessity and sufficiency in turn.

\subsection{Two necessary conditions for correctness}

\begin{condition}\label{cond:1} If $dP$ is correct, i.e., \eqref{correct} holds, it satisfies the SDEs
\be
\int_{\cM_c} (A_i f)dP=0\;,\quad i=1,\ldots,n\quad \forall f\in\cA
\label{sde1}
\ee
and
\be
\int_{\cM_c} (A_{j,k} f)dP=0,\quad j=1,\ldots,d;\quad k=1,\ldots, m \quad \forall f\in\cA\;.
\label{sde2}
\ee
\end{condition}
To prove \eqref{sde1} and \eqref{sde2}, we only have to integrate the right-hand side of \eqref{correct} (i.e., \eqref{eq:rho_expectation}) by parts, which under our
assumptions can be carried out without the appearance of boundary terms, and use the definitions of $A_i$ and $A_{j,k}$ in \eqref{eq:A_i} and \eqref{eq:A_jk}. Note that the validity of \eqref{correct} for all  $f\in\cH$ implies the validity of \eqref{sde1} and \eqref{sde2} only for $f\in\cA$.

The second condition is another straightforward consequence of \eqref{correct}, using the norm $\Vert f\Vert_w$ in 
\eqref{norms} and the constant $C$ defined in \eqref{integrable}:

\begin{condition}\label{cond:2}If $dP$ is correct, i.e., it satisfies \eqref{correct}, the 
following bounds hold for all $f\in\cH$: 
\be 
\left|\int_{\cM_c} f dP\right| \le C \Vert f\Vert_w\;.
\label{bounds}
\ee 
\end{condition}
\begin{proof} 
Using \eqref{correct} and \eqref{eq:rho_split}, we have
\begin{equation}
\left|\int_{\cM_c} f(\z,\g_c) dP\right|
= \left|\int_{\cM_r} f(\x,\g) w(\x)\rho_r(\x,\g) d\mu\right|\;,
\end{equation}
for which we can find a trivial upper bound, namely
\begin{equation}
\left|\int_{\cM_c} f(\z,\g_c) dP\right|
\leq
\left(\int_{\cM_r} \vert\rho_r(\x,\g)\vert d\mu\right) \sup_{\cM_r}  \vert f(\x,\g)\vert w(\x)= C\Vert f\Vert_w\;,
\end{equation}
which is precisely \eqref{bounds}. 
\qedhere
\end{proof}

\subsection{Sufficiency of the necessary conditions}

We now go on to show that \cref{cond:1,cond:2} in combination are also sufficient to guarantee correctness, based on our earlier assumption that
$\rho$ does not vanish anywhere on $\cM_r$.  
To this end, let us assume that \eqref{bounds} holds. This means that the linear functional $T_P$ on $\cH$,
defined by
\be\label{eq:T_P}
(T_P,f)\equiv \int_{\cM_c} f dP\;,
\ee
is bounded in the norm $\Vert\cdot\Vert_w$, or, equivalently, that the linear functional $\tilde T_P$ on
$\tilde\cH\equiv w\cH$, defined by
\be\label{eq:tilde_T_P}
(\tilde T_P,\tilde f)\equiv (T_P,f)\quad  {\rm with} \quad \tilde f=wf\;,
\ee
satisfies
\be
\big\vert(\tilde T_P,\tilde f)\big\vert \le C \Vert\tilde f\Vert_\infty\quad \forall \tilde f\in\tilde \cH\;,
\quad \Vert\tilde f\Vert_\infty\equiv \sup_{\cM_r}\big\vert\tilde f(\x,\g)\big\vert\;.
\label{inf_bound}
\ee
Note that $\tilde \cH \subset\cC_0(\cM_r)$, where $\cC_0(\cM_r)$ 
is the space of all continuous functions on $\cM_r$ vanishing at infinity. This is due to our requirements that $\rho$ does not vanish on $\cM_r$ and that any possible large growth of $f$ at infinity is counteracted by the fast decay of $w$. For the latter to be true, it was necessary to restrict the growth of the partial derivatives of $\log\rho$ earlier.

A crucial fact is the following

\begin{proposition}\label{prop:1} Any linear functional $\tilde{T}$ on $\tilde\cH\subset \cC_0(\cM_r)$ that is bounded 
as in \eqref{inf_bound} can be extended to a linear functional on all of $\cC_0(\cM_r)$ with the same bound.
\end{proposition}

\begin{proof} This follows directly from the \emph{Hahn--Banach} theorem; see, e.g., \cite{Fol99} (Theorem 5.7) or \cite{Rud91} (Theorem 3.3).
\qedhere \end{proof}
 
The extension given by this theorem, which we denote by the same symbol $\tilde{T},$ may not be unique, but we will in the end only be
interested in its action on $\tilde\cH$ anyway.
The usefulness of this extension comes from another abstract result, known as the {\it 
Riesz--Markov} theorem; see, e.g., \cite{Rud87} (Theorem 6.19). Adapting the notation to our case, it says

{\bf Theorem} (Riesz--Markov): {\it If $X$ is a locally compact Hausdorff space, then every 
bounded linear functional $\tilde T$ on $\cC_0(X)$ is represented by a unique regular  complex Borel measure $\nu$ in the sense that $(\tilde T,\tilde{f})=\int_X \tilde{f} d\nu$ for every $\tilde{f}\in 
\cC_0(X)$. Moreover, $\nu$ is of finite total variation, i.e.,
$\Vert\tilde T\Vert\equiv \int_X d|\nu|<\infty$}. 

In our context, $\cM_r$ is locally compact as it is of finite dimension by assumption. Also note that
a Borel measure $d\nu$ on $\cM_r$ always represents a generalized function (actually a tempered distribution, see \cite{RS72}) $\tilde \sigma$, so we can write 
\be
d\nu= \tilde\sigma(\x,\g)d\mu\;.
\ee
Hence, we obtain

\begin{proposition}\label{prop:2}
Any bounded linear functional $\tilde T$ on $\cC_0(\cM_r)$ is given by a 
complex measure, which we represent by a distributional
density  $\tilde\sigma$ on $\cM_r$ as follows:
\be
(\tilde T,\tilde f)=\int_{\cM_r} \tilde\sigma({\x,\g}) \tilde f({\x,\g})\,d\mu\;,\quad \tilde{f} \in \cC_0(\cM_r)\;,
\label{measure}
\ee
where $\tilde\sigma$ satisfies
\be
\int_{\cM_r} |\tilde\sigma| d\mu <\infty\;.
\ee
\end{proposition}

Restricting to the subset $\tilde{\cH}\subset\cC_0(\cM_r)$ and applying this proposition to the functional $\tilde T_P$ that is generated by the probability measure
$dP$ according to \eqref{eq:tilde_T_P} and \eqref{eq:T_P} and that we have just shown to be bounded, we find that it is represented by a complex measure $\tilde \sigma_P$ of bounded total 
variation, 
i.e.,
\be
(\tilde T_P,\tilde f)= \int_{\cM_r} \tilde\sigma_P({\x,\g})\tilde f({\x,\g})\,d\mu\;, 
\quad \tilde f = wf\;, \quad f\in \cH\;.
\label{compl_meas}
\ee
To sum up, we have obtained

\begin{proposition}\label{prop:3}
Let $dP$ be a probability measure on $\cM_c$, satisfying for all $f\in\cH$
\be
\biggl |\int_{\cM_c} f(\z,\g_c) dP \biggr |\le C \Vert f\Vert_w\;,
\ee
with $\Vert f\Vert_w$ and $C$ defined in \eqref{norms} and \eqref{integrable}, respectively, then, again for all $f\in\cH$, we have
\be
\int_{\cM_c} f(\z,\g_c)dP  =\int_{\cM_r}\tilde\sigma_P({\x,\g}) w(\x)f({\x,\g})d\mu\;,
\label{cond_cl}  
\ee
with $\tilde\sigma_P$ of finite total variation.
\end{proposition}
Defining 
\be\label{eq:sigma_P}
\sigma_P(\x,\g)\equiv w(\x) \tilde\sigma_P(\x,\g)\;,
\ee
this can be rewritten as
\be
\int_{\cM_c} f(\z,\g_c)dP  =\int_{\cM_r}\sigma_P({\x,\g})f({\x,\g})d\mu
\label{cond_cl2}\;.  
\ee

In other words, if the probability measure $dP$ on $\cM_c$ is tested only against functions 
in $\cH$, it can be replaced by $\sigma_P$, which is concentrated on $\cM_r$. 
The density $\sigma_P$ decays at least like a Gaussian in $|\x|$ (since $w(\x)$ has the same property and $\tilde{\sigma}_P$ is of finite total variation) for all $\g$. 
We now show that \cref{cond:1} (i.e., the SDEs \eqref{sde1} and \eqref{sde2}) together with \cref{cond:2} (i.e., the bounds \eqref{bounds}) implies 
correctness.

\begin{proof}
Integration by parts of \eqref{sde1} and \eqref{sde2}, after using \eqref{cond_cl2}, rolls the respective SD operators over to
$\sigma_P$, changing them into their transposes as
\be
\int_{\cM_r} f({\x,\g}) A_i^T \sigma_P(\x,\g)\,d\mu=0
\quad \forall f\in\cA\;, \quad i=1,\ldots, n
\label{int_byp1}
\ee
and
\be
\int_{\cM_r} f({\x,\g}) A_{j,k}^T \sigma_P(\x,\g)\,d\mu=0
\quad  \forall f\in\cA\;, \quad j=1,\ldots, d;\, k=1,\ldots,m\;,
\label{int_byp2}
\ee
where
\be
A_i^T=- \partial_{i} + \frac{\partial_{i}\rho}{\rho}
\qquad
\textnormal{and} \qquad
A_{j,k}^T=-D_{j,k}+\frac{D_{j,k}\rho}{\rho}\;.
\ee
Note that under our conditions, there are once again no boundary terms here. Up to now, we have treated $\sigma_P$ as a distribution, in fact one of order zero \cite{Rud91} according to the bounds \eqref{bounds}. $\sigma_P$ therefore represents a complex (Borel) measure by the Riesz-Markov theorem.
But we have not yet addressed the question whether it is a differentiable function, so the derivatives in \eqref{int_byp1} and \eqref{int_byp2} have to be understood in the distributional sense.

We would like to conclude from \eqref{int_byp1} and \eqref{int_byp2} that (for all $i$, $j$, and $k$)
\be\label{eq:sde_meas}
A_i^T \sigma_P=0 \quad \textnormal{and} \quad A_{j,k}^T \sigma_P=0\;.
\ee
These differential equations then would imply that $\sigma_P$ is really a differentiable function. Indeed, this will turn out to be true because there are enough functions in $\cA$ such that a distribution $T$ 
satisfying $(T,f)=0\;\forall f\in\cA$ vanishes identically. However, in order to see this more explicitly, we require a few more arguments.\\

We consider functions of the form
\be
f_\bk(\x,\g)= \exp(\ii\bk\cdot\x) h(\g)\;,
\label{functionform}
\ee
where $h$ is in the sub-algebra $\cA_G$ of $\cA$ consisting of functions that do not depend on 
$\x$. Note that the functions $f_\bk$ are not in $\cA$, but they are contained in its closure $\overline\cA$ with respect to $\Vert\cdot\Vert_w$, as is shown in \cref{app:closure}.

Now, to prove \eqref{eq:sde_meas}, we consider the functions
\be
\hat F_i(\bk)=\int_{\cM_r}\sigma_P(\x,\g) A_i f_\bk(\x,\g)\,d\mu
\ee
and
\be
\hat F_{j,k}(\bk)=\int_{\cM_r}\sigma_P(\x,\g) A_{j,k}f_\bk(\x,\g)\,d\mu\;,
\ee
which all vanish due to \eqref{int_byp1}, \eqref{int_byp2},
and the fact that $f_\bk$ is in $\bar\cA$.
Moreover, $\hat F_i(\bk)$ is the Fourier transform of 
\be
F_i(\x)=\int_{G^m} h(\g)A_i^T\sigma_P(\x,\g) \,d^mg\;,
\label{sde_inter1}
\ee
and $\hat F_{j,k}(\bk)$ the Fourier transform of 
\be
F_{j,k}(\x)=\int_{G^m} h(\g)A_{j,k}^T\sigma_P(\x,\g) \,d^mg\;,
\label{sde_inter2}
\ee
which also all vanish because the Fourier transformation is one-to-one on the space of (tempered) distributions on $\R^n$ 
\cite{RS72}.
By the \emph{Peter--Weyl} theorem (see, e.g., \cite{Bum13} (Theorem 4.1) or
\cite{Kna86} (Theorem 1.12)), $\cA_G$ is dense in the space $\cC(G^m)$ of
continuous functions on $G^m$. Therefore, it is also dense in the space $\cC^\infty(G^m)$, the
standard space of test functions on $G^m$, in the appropriate topology. As a consequence, the vanishing
of \eqref{sde_inter1} and \eqref{sde_inter2} for arbitrary $h(\g)$ implies \eqref{eq:sde_meas}. 

Owing to the requirement that $\rho$ does not vanish anywhere on $\cM_r$, the differential equations \eqref{eq:sde_meas} have the general solution
\be
\sigma_P(\x,\g)= c \rho(\x,\g)
\ee
with $c\in\C$. By choosing $f=\1$ and from the fact that $P$ is normalized, we find that $c=1$ and 
correctness \eqref{correct} holds.
\qedhere
\end{proof}

Let us summarize what we have shown:

\begin{theorem}\label{thm:1}
Let $P$ be a probability distribution on $\cM_c$. Assume that
$\rho$ does not vanish on $\cM_r$ and that all partial derivatives of $\log\rho$ grow at most like a polynomial for $\vert\x\vert\to\infty$. Then, the following are equivalent:
\begin{enumerate}[(a)]
    \item $P$ is  correct in the sense that for any  $f\in\cH$
    \be
    \int_{\cM_c} f dP = \int_{\cM_r} f \rho\,d\mu\;.
    \ee
    \item $P$ satisfies the SDEs \eqref{sde1} and \eqref{sde2} as well as the bounds
    \be
    \left|\int_{\cM_c} f dP\right|\le C \Vert f\Vert_w\;,
    \ee
    with $\rho(\x,\g)=w(\x) \rho_r(\x,\g)$, $w$ obeying \eqref{eq:w_decay},
    \be
    \Vert f\Vert_{w}\equiv \sup_{\cM_r} |f(\x,\g)|w(\x)<\infty  \quad \forall f\in\cH\;,
    \ee 
    and
    \be
    C\equiv \int_{\cM_r} |\rho_r(\x,\g)| d\mu<\infty\;.
    \ee
\end{enumerate}
\end{theorem}

{\bf Remark:} Up to now, we have tried to keep things as general as possible. As a consequence, 
\cref{thm:1} does not provide one single condition for correctness, but in fact an entire family of 
conditions: For one, the choice of the real manifold $\cM_r$ is not unique, as, for instance, continuous deformations (with fixed endpoints) of $\cM_r$ give rise to the same ``correct expectations'' \eqref{eq:rho_expectation} due to Cauchy's theorem and the assumed holomorphicity of $f$ and $\rho$.
However, the bounds $C\Vert f\Vert_w$ do depend on the choice of $\cM_r$ (keep in mind, though, that $\cM_r$ must always be chosen such that it contains no zeros of $\rho$). In addition, the decomposition of $\rho$ into $w$ and $\rho_r$ in \eqref{eq:rho_split} is also arbitrary as long as the resulting $C$ and $\Vert f\Vert_w$ ($\forall f\in\cH$) are finite and $w$ decays at least like a Gaussian.
The constant $C$ and the norm $\Vert f\Vert_w$, as well as their product, will, of course, differ for different decompositions. 

What is perhaps even more important is the freedom in the choice 
of norms. We have used in \eqref{norms} and \eqref{integrable}
the supremum norm for the observables and the $L^1$ norm for the reduced density $\rho_r$. One may, however, just as well choose, for instance, $L^p$ and $L^q$
norms with $1/p+1/q=1$ instead. In particular, one could
choose $p=q=2$, i.e., 
$L^2$ norms for both quantities, replacing 
\eqref{norms} and \eqref{integrable} by
\begin{equation}\label{eq:f_norm_general}
\Vert f\Vert_{w}\equiv \sqrt{\int_{\cM_r} |f(\x,\g)|^2 w(\x)^2 d\mu}<\infty  \quad \forall f\in\cH
\end{equation}
and
\begin{equation}\label{eq:C_general}
C\equiv \sqrt{\int_{\cM_r} |\rho_r(\x,\g)|^2 d\mu}<\infty\;,  
\end{equation}
respectively. It should be understood
that the practical applicability of the proposed correctness criterion might strongly depend on all of 
these choices. A detailed analysis systematically comparing different possible correctness conditions arising from our proposal is, however, beyond the scope of this work.
\section{Basics of complex Langevin simulations}\label{sec:cle}
    Having developed a novel correctness criterion for probability distributions $P$ on complexified
manifolds $\cM_c$, a logical next step is to put it to the test in practical applications. Our
main interest is its applicability in the context of complex Langevin simulations, which we shall thus
briefly motivate and introduce in the present section.

\subsection{The sign problem}
In the path integral picture of quantum field theory, expectation values of observables $\cO$ are
defined as integrals of the form
\begin{equation}\label{eq:qft_expectation_values}
    \langle\cO\rangle_\rho \equiv \int \cD\Phi \cO[\Phi]\rho[\Phi]\;,
\end{equation}
where $\Phi$ collectively denotes the dynamical degrees of freedom of the theory under consideration, $\cD\Phi$ is the
appropriate path integral measure for these fields and $\rho[\Phi]$ is some density giving a
nontrivial weight to each field configuration. Typically, 
\begin{equation}
    \rho[\Phi] = \frac{e^{\ii S_M[\Phi]}}{\int\cD\Phi e^{\ii S_M[\Phi]}} 
        \quad \textnormal{or} \quad 
    \rho[\Phi] = \frac{e^{-S_E[\Phi]}}{\int\cD\Phi e^{-S_E[\Phi]}}\;,
\end{equation}
for a theory with an action $S_M[\Phi]$ in Minkowski spacetime, or, (after performing a suitable Wick rotation) an action $S_E[\Phi]$ in Euclidean space, respectively. The normalization
factors in the denominator ensure that $\langle\1\rangle_\rho=1$. 

The most successful non-perturbative tool for computing $\langle\cO\rangle$ from first principles is
lattice quantum field theory, in which -- conventionally -- the integral in 
\eqref{eq:qft_expectation_values} is approximated by a finite sum over configurations $\{\Phi\}_i$ 
distributed according to $\rho[\Phi]$, which is the traditional importance sampling approach. For this to make sense mathematically, however, one requires that $\rho[\Phi]$ 
is real and nonnegative, i.e., that it can be interpreted as a probability density. Clearly, this
is not the case in Minkowski spacetime, which is why in the vast majority of lattice approaches one generically works in Euclidean space. 
However, the Euclidean action $S_E$ is not guaranteed to be real either. Cases in which $S_E$ 
acquires a nonzero imaginary part include (but are not restricted to)  quantum field theories (in thermodynamic equilibrium) at
nonvanishing chemical potential or with topological terms. 

In all of the aforementioned cases, the 
straightforward application of usual lattice techniques based on importance sampling fails. This is the 
infamous sign (or complex-action) problem \cite{TW05}; see, e.g., \cite{For10pr,GL16,BRL21} for reviews.

\subsection{The complex Langevin equation}
One particular method that is aimed at circumventing the sign problem is the complex Langevin 
approach \cite{Par83,Kla83} based on stochastic quantization \cite{PW81}. The underlying
mechanism is the stochastic evolution of complexified field degrees of freedom in a fictitious
time dimension $\tau$, governed by the complex Langevin equation. For simplicity,
let us restrict to $N$ real scalar degrees of freedom $\Phi_i^R$ ($i=1,\dots,N$) described by the
Euclidean\footnote{Here and in the following, we work in Euclidean space and suppress the 
subscript `$E$' on the Euclidean action. Note, however, that the discussion can straightforwardly
be transferred to Minkowski spacetime via the replacement $-S_E\to \ii S_M$.} action 
$S[\Phi^R]$. Then, the complex Langevin equation reads\footnote{Implicit summation over indices occurring twice in a product is assumed throughout.} 
\begin{equation}\label{eq:cle}
    d\Phi_{i}(\tau) = D_{i}(\Phi)d\tau + H_{ij}dw_j(\tau)\;,
\end{equation}
where the $\Phi_i\equiv\Phi^R_i+\ii\Phi^I_i$ result from the continuation of each $\Phi_i^R$ into the complex plane. Moreover, $D_{i}$ denotes the so-called drift term,
\begin{equation}
    D_{i} = -H_{ik}H_{jk}\frac{\partial S}{\partial\Phi_j} +
    \frac{\partial (H_{ik}H_{jk})}{\partial\Phi_j}\;,
\end{equation}
$H_{ij}$ is a (generally $\Phi$-dependent) complex $N\times N$ matrix called the kernel and $dw_i$ describes a standard Wiener process in each degree of freedom, such that
\begin{equation}
    \langle dw_j(\tau)\rangle = 0\;, \quad \langle dw_j(\tau)dw_k(\tau)\rangle = 2\delta_{jk}d\tau
\end{equation}
for all field indices $j$ and $k$.
In order for this approach to be well defined, the action must allow for an analytic continuation
$S[\Phi]$ in its arguments and the same must be true for the observables 
$\cO[\Phi]$. We also mention that the extension of the Langevin equation to gauge theories is relatively straightforward; see, e.g., \cite{BKK85}.

Now, the stochastic evolution governed by the complex Langevin equation gives rise to a (real and nonnegative) probability distribution 
$P(\Phi^R_i,\Phi^I_i;\tau)$ on the complex manifold, whose $\tau$-dependence is, in turn,
determined by an associated so-called Fokker--Planck equation. The exact form of the latter, however, is irrelevant for the
present discussion. In the following, we assume that $P$ has an equilibrium limit
\begin{equation}
    P(\Phi^R_i,\Phi^I_i) \equiv \lim_{\tau\to\infty}P(\Phi^R_i,\Phi^I_i,\tau)\;.
\end{equation}
Note that this assumption is not always justified, but only in the case when there exists such a limit is the complex
Langevin approach applicable. As a consequence, the complex Langevin evolution after equilibration indeed leaves us with a probability distribution $P$ on a 
complex manifold $\cM_c$, for which the correctness condition developed in \cref{sec:criteria} can be checked. Before doing so below in the context of a few simple models, however, let us first recall 
possible reasons for why and how complex Langevin simulations may fail in the first place.

\subsection{Reasons of failure for the complex Langevin approach}
The study of possible error sources for complex Langevin simulations has a rather long history
\cite{Sal93,GL93,FOS94,Gau98,ASS10,AJS10,AJ10,AJS11,AJP13,AGS13,NS15,Sal16,NNS16,ASS17,SS19,SSS19,SSS20,Sei20,SSS24,BHM25}. One common issue is the presence of so-called runaway trajectories,
i.e., paths in complexified field space along which the classical evolution (i.e., the one excluding the noise
term in \eqref{eq:cle}) diverges. Their existence may potentially lead to large discretization effects or even cause simulations to crash
when the involved drift terms become too large. However, the effects of runaway trajectories can often
be mitigated or even eliminated entirely in practice by the use of adaptive step-size control \cite{AJS10}, implicit solvers \cite{ALR21},
and (in the case of gauge theories) the gauge cooling technique \cite{SSS13}. At the very least,
it is usually possible to detect whether or not runaways have an impact on a given simulation by, for 
instance, monitoring the distance of the simulation trajectory from the real manifold.

In addition, complex Langevin simulations may suffer from a loss of ergodicity, especially
in the presence of poles in the drift term. In this case, one may find that the sampled configuration space is split into different disconnected parts \cite{FOS86,FOS94,ASS17,SS19}. While transitions between these regions are
possible in simulations due to the finiteness of the discrete Langevin 
step size, they become increasingly rare in the continuum limit and their absence can lead to incorrect results. It is not entirely clear how to deal with a loss of ergodicity in the general case, even though in a few models there exist solutions employing, for instance, a kernel or complex noise \cite{ASS17}.

Another issue concerns the decay properties of the probability distribution $P$ resulting from the 
complex Langevin evolution. Generically, one expects the complex Langevin results to be correct only if 
$P$ decays fast enough in all complex directions. More precisely, a correctness criterion for the complex
Langevin evolution was formulated in \cite{ASS10,AJS11} that relies on the sufficiently fast 
decay of $P\cO$ ($\cO$, as before, denoting an observable) such that one may perform an integration by parts without the appearance of 
boundary terms. Most generally, however, one cannot expect such boundary terms to vanish, which is why
it is advisable to measure them in practice \cite{SSS19,SSS20}. If the boundary terms are nonzero, the obtained
simulation results are generally incorrect (even though for some observables the expectation values may still come out right). 

The converse, however, is not always true, i.e., the absence of boundary terms does not 
automatically imply the correctness of results \cite{ASS17,SS19,HMS25}. Similarly, while the validity of the SDEs is necessary for correct convergence, it is not sufficient.
In particular, as was shown 
analytically for one-dimensional systems in \cite{SS19}, any linear functional $T$ (such as the one provided by the complex Langevin evolution) on a suitable space of test functions (such as the algebra $\cA$ defined in \cref{sec:generalities}) that obeys the SDEs
is given by a linear combination of the form
\begin{equation}\label{eq:sase_theorem}
    (T,f) = \sum_{i=1}^{n_\gamma} a_i(T_{\gamma_i},f)\;.
\end{equation}
Here, the $a_i$ are complex coefficients and 
\begin{equation}
    (T_{\gamma_i},f) \equiv \int_{\gamma_i} dz f(z)\rho(z)
\end{equation}
are the ``expectation values'' of $f$ along paths $\gamma_i$, which are linearly independent so-called integration cycles (in the 
terminology of \cite{Wit11}). Concretely, an integration cycle $\gamma_i$ is either a path that connects two distinct zeros of $\rho(z)$ or a closed noncontractible path.
It is important to note that by an integration cycle one really means an equivalence class of  integration paths that can be deformed into one another or are equivalent via Cauchy's theorem. In \eqref{eq:sase_theorem}, we have denoted the total number of 
linearly independent integration cycles by $n_\gamma$. Now, if we define (without loss of generality) $\gamma_1=\cM_r$, we say that results are correct if and only if $a_i=\delta_{i1}$. 

The role of integration cycles in the context of 
complex Langevin simulations was studied in detail in \cite{HMS25}, where evidence for the 
validity of \eqref{eq:sase_theorem} beyond one dimension was established as well. 
We can summarize the situation as follows: if for a particular complex Langevin simulation boundary terms 
vanish/the SDEs hold, but the results are nonetheless incorrect, this is necessarily due to 
contributions from unwanted integration cycles, i.e., from $\gamma_i\neq \cM_r$. The relationship between vanishing boundary terms, the validity of the SDEs and the so-called complex Langevin consistency conditions \cite{AJS11} is discussed in some detail in \cref{app:cc_sde}.
We mention in passing that the sampling of unwanted integration cycles can in some cases be
related to the aforementioned ergodicity problems \cite{ASS17,SS19}. Either way, the 
correctness criterion we propose in this work should -- at least in principle -- be able to fill the gap left open by
the boundary-term (and equivalent) analyses, enabling one to detect incorrect convergence even when boundary terms 
are zero. 

Let us now briefly address the scope of applicability of this correctness criterion in
practice. Needless to say, there is little hope in trying to use it to establish unambiguously that
a particular simulation gives rise to correct results. Since this would require measuring every
observable contained in the space $\cA$ (or $\cH$), which is of infinite size, it is clearly an impossible task.
What it should allow one to do, however, is rule out certain results despite the absence
of boundary terms or in cases where boundary terms are very small or noisy. To exemplify this idea, below the
criterion will be applied in the context of a few simple models for which complex Langevin simulations are known to
produce no boundary terms but whose results nonetheless disagree with exact solutions \cite{ASS17,SS19,HMS25}. Obviously, in most realistic systems no exact results are available, making any sort of correctness criterion all the more desirable.

\section{Examples and numerical checks}\label{sec:numerics}
    In this section, we perform numerical tests on the new condition for correctness for a handful of simple
systems with known results. We restrict ourselves to models with $m=0$, i.e., there are no Lie-group-valued degrees of freedom, such that $\cM_r=\R^n$. A more detailed analysis in the case of gauge theories is deferred to future work. Following the FAIR guiding principles \cite{FAIR16}, we make the simulation data underlying the results presented in this section \cite{data}, as well as the used analysis scripts \cite{scripts}, available online. Moreover, our simulation code is available upon request.

\subsection{Regular models}
We begin by considering models whose densities are given by $\rho(\x)=e^{-S(\x)}$ with polynomial $S(\x)$. This simplifies the discussion as in this case $\cA=\cH$, i.e., we can restrict to an algebra of polynomial observables. The numerical simulation and data analysis setup we use for these models is outlined in detail in \cite{HMS25}. In fact, we re-use parts of the data set generated for \cite{HMS25} in the present work.

\subsubsection{One dimension}
The first model to be investigated is given by the action
\begin{equation}\label{eq:quartic_1d}
    S(x) = \frac{\lambda}{4}x^4\;,
\end{equation}
with a single real degree of freedom $x$ and a complex parameter $\lambda$ obeying 
$\vert\lambda\vert=1$. The correctness of the
probability distribution $P$ produced by complex Langevin simulations including a kernel in this 
model was first
investigated in \cite{OOS89}. More recently, in \cite{MHS24p,HMS25} the model was revisited and the 
connection between the kernel and the relevant integration cycles was investigated. In particular, 
it was discovered that a kernel can -- to a certain extent -- be used to control which integration cycles 
are being sampled in a simulation. Here, we once again revisit this simple model as a testbed for 
our correctness criterion. 

The complex Langevin equation \eqref{eq:cle} for the theory \eqref{eq:quartic_1d}, after complexification 
$x\to z=x+\ii y$, simplifies to
\begin{equation}
    dz(\tau) = -H^2S'(z)d\tau + Hdw(\tau)\;,
\end{equation}
where we have assumed the kernel $H$ to be independent of $z$ for simplicity. It is thus merely a constant complex number that we may choose at will. As in \cite{OOS89,MHS24p,HMS25}, we parametrize the kernel\footnote{In the cited works, it was $K\equiv H^2$ that was parametrized instead of $H$, but here all definitions are adapted accordingly to produce an equivalent outcome.} with an integer $m_H$ as
\begin{equation}\label{eq:kernel_quartic_1d}
    H = e^{-\ii \frac{m_H}{48}}\;.
\end{equation}
An important finding of \cite{MHS24p,HMS25} was that by an appropriate choice of kernel one can indeed obtain $a_i=\delta_{i1}$ in \eqref{eq:sase_theorem}, where $\gamma_1=\cM_r$. If any cycle other than $\cM_r$ were to give a sizable contribution to the results, the latter would be considered incorrect.
This incorrectness, however, could not be detected simply via a boundary-term analysis. In the 
following, however, we shall demonstrate that it can be detected via the new criterion. 

To this end, we make use of the freedom of choosing $\cM_r$ at our discretion. In 
particular, since the real $x$ axis does not correspond to a valid integration cycle if
$\mathrm{Re}(\lambda)<0$, we define as $\cM_r$ the integration contour resulting from rotating the real axis as $x=\lambda^{-1/4}\xi$, along which the
analytically continued action $S(\xi)$ is real. We mention, however, that
in a realistic theory the choice of $\cM_r$ is dictated by the physics and one only has the freedom of continuously deforming it, contrary to what is done here. As is detailed in \cite{HMS25},
there are two more independent integration cycles in this model. In analogy to \cite{HMS25}, we choose the imaginary axis
and a contour connecting $z=\ii\infty$ with $z=\infty$, both after the same rotation as for the real axis. In the following, we set $\lambda=e^{5\ii\pi/6}$. Moreover, from now on we will exclusively use the letter $z$ to denote the dynamical degree of freedom to avoid possible confusion arising from jumping back and forth between $x$ and $z$.

In order to test for correctness with the new criterion, we need to establish the validity of two requirements: the SDEs \eqref{sde1} and the bounds \eqref{bounds}. For the former, we consider monomial observables $z^k$ with integer $k>0$, for which the SDEs read 
\begin{equation}\label{eq:dse_quartic_1d}
    \left\langle Az^k\right\rangle = k\left\langle z^{k-1}\right\rangle - \lambda\left\langle z^{k+3}\right\rangle = 0\;,
\end{equation}
with $A$ given in \eqref{eq:A_i}. The left-hand side of this equation, as resulting from complex Langevin simulations (with complex Langevin expectation values labelled by the subscript ``CL''), is plotted as a function of the kernel parameter $m_H$ in \cref{fig:quartic_1d} (left). Notice that we only show results for those $m_H$ for which there are no boundary terms \cite{HMS25}. In the cases where boundary terms are nonzero, we found the SDEs to be strongly fluctuating, making it impossible to obtain a usable signal. Also note that fluctuations increase with increasing $k$ such that we refrain from showing results for $k>5$. Both of these observations, however, are consistent with the expectations. 

The main message of \cref{fig:quartic_1d} (left) is that there are four different regions in $m_H$ for which the SDEs are well satisfied and boundary terms vanish. Nonetheless, it is known from \cite{OOS89} that only one of these regions corresponds to correct solutions. Thus, we are indeed confronted with a situation where the study of boundary terms (or SDEs) alone is not sufficient to make definite statements about correct convergence. This was found to be the case also in another simple model in \cite{ALR23}. The same is true for the widely-employed correctness criterion based on the decay properties of the complex-Langevin drift term \cite{NNS16}. Indeed, as was shown in \cite{HMS25}, the respective distributions of $z$ in the complex plane for, e.g., $m_H=10$ and $m_H=34$ are identical up to a rotation. This can be understood from the fact that for both these choices of kernel the complex Langevin evolution is equivalent to a real Langevin evolution, albeit on a rotated real axis, the rotation angle being the only thing distinguishing between them. As a consequence, the distributions of drift magnitudes must be equivalent by definition. Clearly, thus, the drift criterion cannot detect this incorrect convergence. Let us therefore invoke the new correctness criterion.

\begin{figure}[t]
    \begin{subfigure}{0.49\linewidth}
        \centering
        \includegraphics[width=\textwidth]{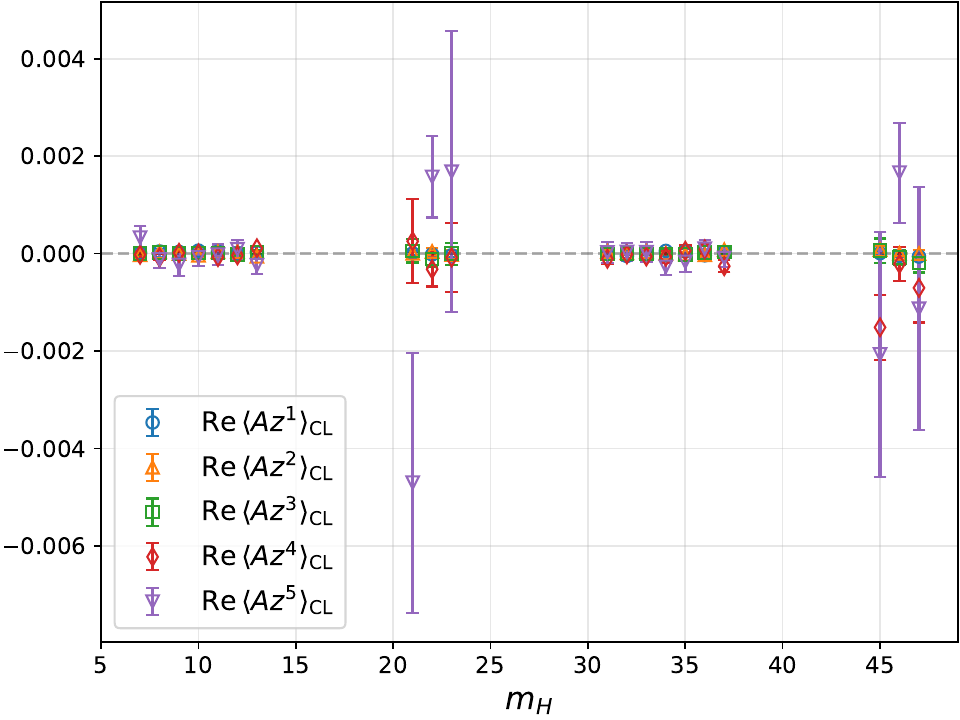}
    \end{subfigure}
    \begin{subfigure}{0.49\linewidth}
        \centering
        \includegraphics[width=\textwidth]{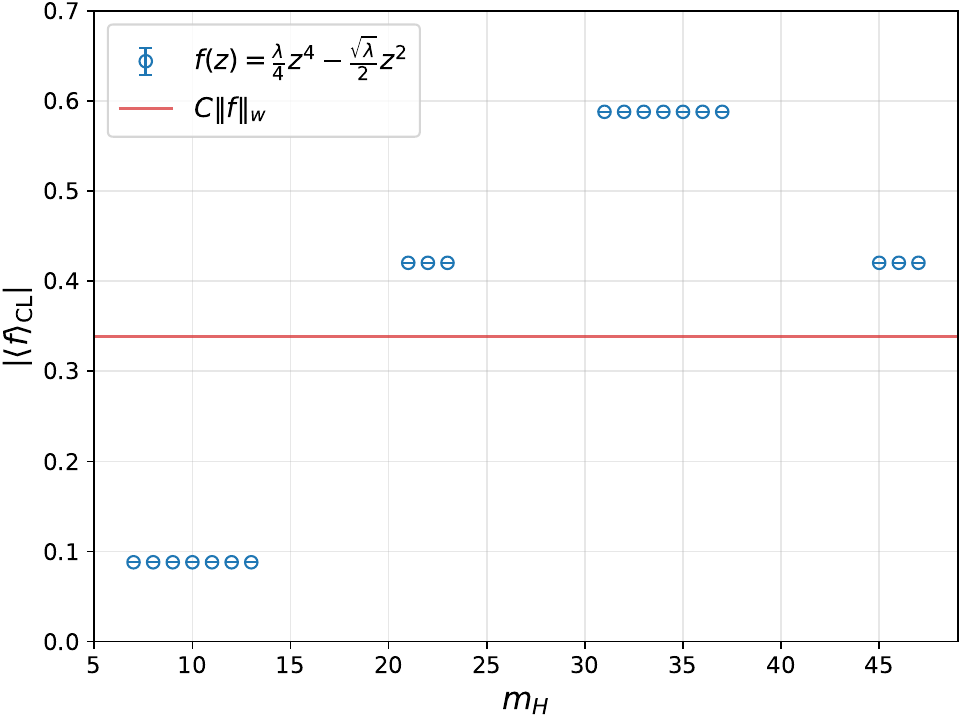}
    \end{subfigure}
    \caption{Complex Langevin results for the model \eqref{eq:quartic_1d}, with $\lambda=e^{5\ii\pi/6}$ and a kernel of the form \eqref{eq:kernel_quartic_1d}, as a function of the parameter $m_H$. We only show results for values of $m_H$ for which boundary terms are consistent with zero. \emph{(left):} Validity of the SDEs \eqref{eq:dse_quartic_1d} with $k=1,\dots,5$. Note that we only plot the real parts of $\left\langle Az^k\right\rangle$ here; the imaginary parts show similar behavior. The dashed horizontal line marks zero. \emph{(right):} Absolute value of the expectation value of the control observable $f(z)$ in \eqref{eq:control_variable_quartic_1d}. The solid horizontal line shows the bound $C\Vert f\Vert_w$ from \eqref{bounds}.}
    \label{fig:quartic_1d}
\end{figure}

To this end, we follow \cref{sec:measure_correctness} by splitting
\begin{equation}
    \rho(z)=\frac{1}{Z}e^{-S(z)}\;, \quad Z=\int_{\cM_r}dz\,e^{-S(z)}
\end{equation} 
according to \eqref{eq:rho_split}, with the choice
\begin{equation}
    w(z) = e^{-\frac{S(z)}{2}}\;, \quad  \rho_r(z) = \frac{1}{Z}e^{-\frac{S(z)}{2}}\;.
\end{equation}
With this, the constant $C$ in \eqref{integrable} is simply given by
\begin{equation}
    C=\frac{1}{\vert Z\vert}\int_{\cM_r}dz \left\vert e^{-\frac{S(z)}{2}}\right\vert=
    \frac{\int_{-\infty}^{\infty}d\xi e^{-\frac{\xi^4}{8}}}{\left\vert\int_{-\infty}^{\infty}d\xi e^{-\frac{\xi^4}{4}}\right\vert}=\sqrt[4]{2}\;,
\end{equation}
which is independent of $\lambda$ due to our $\lambda$-dependent choice of $\cM_r$. Moreover, the norm $\Vert f\Vert_w$ in \eqref{norms},
\begin{equation}
    \Vert f\Vert_w = \sup_{z\in\cM_r}\vert f(z)\vert e^{-\frac{S(z)}{2}}\;,
\end{equation}
can straightforwardly be evaluated numerically for polynomial $f(z)$.
Let us consider the following observable:
\begin{equation}\label{eq:control_variable_quartic_1d}
    f(z) = \frac{\lambda}{4}z^4-\frac{\sqrt{\lambda}}{2}z^2\;,
\end{equation}
for which we find $C\Vert f\Vert_w\approx0.34$ for the bound on the right-hand side of \eqref{bounds}. In \cref{fig:quartic_1d} (right), we show the
$m_H$-dependence of the complex Langevin expectation value of $f$ for those values of $m_H$ for which the SDEs are valid. Moreover, $C\Vert f\Vert_w$ is shown as 
the full horizontal line for comparison. 

One finds that out of the four regions on which the SDEs hold, the bound \eqref{bounds} is obeyed only on one, namely the one around $m_H=10$. This is precisely the range of $m_H$ on which the complex Langevin evolution produces correct results \cite{OOS89}, i.e, samples only the desired integration cycle and no other \cite{HMS25}. Thus, with this 
particular choice of observable, the proposed correctness criterion is capable of ruling out the other
plateaus. We emphasize once again, however, that the results of \cref{fig:quartic_1d} (right) are not sufficient on their own to prove correctness on the plateau around $m_H=10$. To unambiguously prove correctness, one would have to check an infinite number of bounds. 

\subsubsection{Two dimensions}
A very similar analysis can be carried out for a straightforward generalization of 
\eqref{eq:quartic_1d} to two dimensions that was also studied in \cite{HMS25}:
\begin{equation}\label{eq:quartic_2d}
    S(x_1,x_2) = \frac{\lambda}{4}(x_1^2+x_2^2)^2\;.
\end{equation}
As discussed in \cite{HMS25}, there are only two independent integration cycles in this model. For real and positive $\lambda$, they correspond to the two manifolds on which, after complexification, $z_1$ and $z_2$ are both real and both imaginary, respectively. For general $\lambda$, we rotate these cycles similarly to the one-dimensional case discussed above. In particular, the cycle $\gamma_1=\cM_r$ is defined such that on $\gamma_1$ the action is real. Once again, in the following we use the notation $z_i$ for the degrees of freedom exclusively.

The Langevin equation we employ reads
\begin{equation}
    dz_i(\tau) = -H^2\frac{\partial S}{\partial z_i}d\tau + Hdw_i(\tau)\;,
\end{equation}
i.e., we assume the kernel to be proportional to the identity matrix, $H_{ij}=\delta_{ij}H$, and we parametrize $H$ as in \eqref{eq:kernel_quartic_1d}. 
For testing the SDEs, we consider monomial observables of the form $z_1^{k_1}z_2^{k_2}$ with integer $k_i\geq0$. The SDEs \eqref{sde1} then read
\begin{align}\label{eq:dse_quartic_2d}
    \begin{aligned}
        \left\langle A_1z_1^{k_1}z_2^{k_2}\right\rangle = k_1\left\langle z_1^{k_1-1}z_2^{k_2}\right\rangle - \lambda\left[\left\langle z_1^{k_1+3}z_2^{k_2}\right\rangle + \left\langle z_1^{k_1+1}z_2^{k_2+2}\right\rangle\right] = 0\;,\\
        \left\langle A_2z_1^{k_1}z_2^{k_2}\right\rangle = k_2\left\langle z_1^{k_1}z_2^{k_2-1}\right\rangle - \lambda\left[\left\langle z_1^{k_1+2}z_2^{k_2+1}\right\rangle + \left\langle z_1^{k_1}z_2^{k_2+3}\right\rangle\right] = 0\;.
    \end{aligned}
\end{align}
The left-hand sides of \eqref{eq:dse_quartic_2d} as a function of the kernel parameter $m_H$ for the simple cases $(k_1,k_2)=(0,1)$ and $(1,0)$ are shown in \cref{fig:quartic_2d} (left). Again, we find different extended regions in $m_H$ on which the SDEs hold and boundary terms vanish (for values of $m_H$ where there are nonzero boundary terms again no results are shown). Here, however, there are only two such regions. This has to do with the way the number of independent integration cycles is reduced from three in one dimension down to two in the two-dimensional model \eqref{eq:quartic_2d} and is discussed in \cite{HMS25}. 

\begin{figure}[t]
    \begin{subfigure}{0.49\linewidth}
        \centering
        \includegraphics[width=\textwidth]{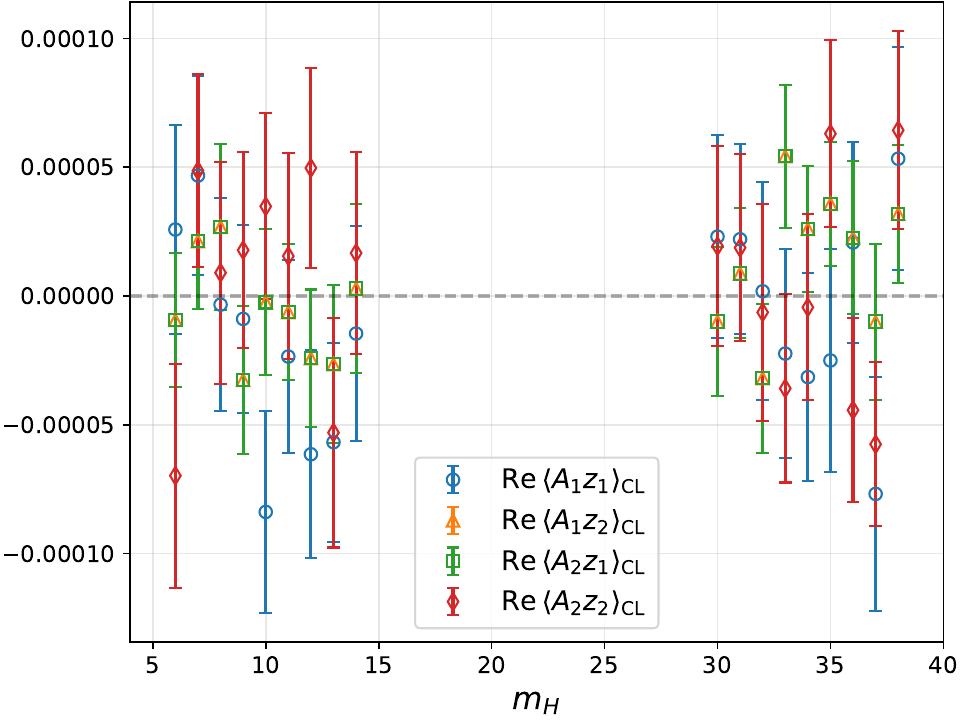}
    \end{subfigure}
    \begin{subfigure}{0.49\linewidth}
        \centering
        \includegraphics[width=\textwidth]{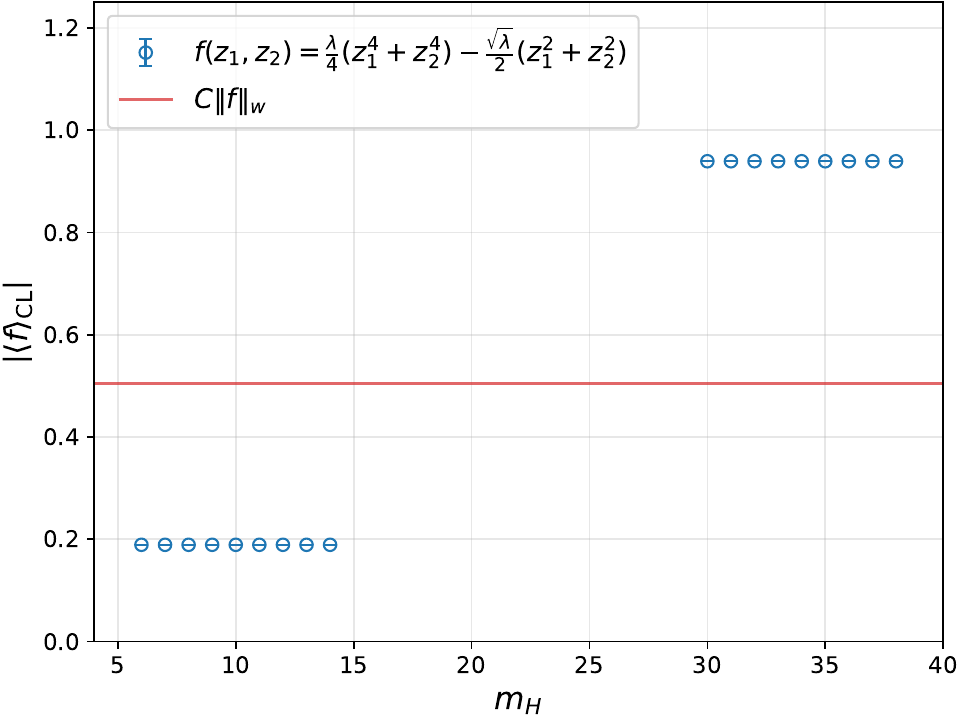}
    \end{subfigure}
    \caption{Similar to \cref{fig:quartic_1d} but for the model \eqref{eq:quartic_2d}. We again use $\lambda=e^{5\ii\pi/6}$ and a kernel proportional to the identity matrix, its entries given by \eqref{eq:kernel_quartic_1d}. \emph{(left):} Validity of \eqref{eq:dse_quartic_2d} with $k_1=0$, $k_2=1$ and vice versa. \emph{(right)}: Expectation value of $f(z)$ in \eqref{eq:control_variable_quartic_2d} and comparison with $C\Vert f\Vert_w$ from \eqref{bounds}.}
    \label{fig:quartic_2d}
\end{figure}

In order to apply the correctness criterion, we split $\rho(z_1,z_2)$ into
\begin{equation}
    w(z_1,z_2) = e^{-\frac{S(z_1,z_2)}{2}}\;, \quad \rho_r(z_1,z_2)=\frac{1}{Z}e^{-\frac{S(z_1,z_2)}{2}}\;,
\end{equation}
with
\begin{equation}
    Z=\int_{\cM_r}dz_1dz_2\,e^{-S(z_1,z_2)}\;,
\end{equation}
similar to before, and we
obtain $C=\sqrt{2}$ from \eqref{integrable}. As a control observable, we choose
\begin{equation}\label{eq:control_variable_quartic_2d}
    f(z_1,z_2) = \frac{\lambda}{4}\left(z_1^4+z_2^4\right)-\frac{\sqrt{\lambda}}{2}\left(z_1^2+z_2^2\right)\;,
\end{equation}
yielding $C\Vert f\Vert_w\approx0.50$. The analog of \cref{fig:quartic_1d} (right) for the two-dimensional model is shown in \cref{fig:quartic_2d} (right), with equivalent conclusions.

We have thus established that the new correctness criterion can detect incorrect convergence due to
unwanted integration cycles in simple quartic models in one and two dimensions. We believe that 
these results can be extended to quartic models in higher dimensions without too many complications. Regarding models with more involved polynomial actions, we are confident that one should be able to find viable control observables $f(\z)$ reasonably easily, but we have not attempted to do so yet. We stress that $f(z)$ and $f(z_1,z_2)$ above were chosen in a purely \emph{ad-hoc} manner and not derived systematically. At the time of writing, we are not aware of any good guiding principle on how to choose appropriate $f(\z)$ for general theories. This will become particularly clear for the model studied in the next subsection.

\subsection{Densities with zeros}\label{sec:density_zeros}
We now turn to a (one-dimensional) theory whose density $\rho$ has a zero on $\cM_c$. In particular, we shall be concerned with the model described by 
\begin{equation}\label{eq:one_pole}
    \rho(z) = \frac{1}{Z}(z-z_p)^{n_p}e^{-\beta z^2}\;, \quad Z = \int_{-\infty}^{\infty}dz(z-z_p)^{n_p}e^{-\beta z^2}\;.
\end{equation}
This theory, for which again $\cM_r=\R$, was studied before in 
\cite{ASS17,SS19}. In particular, it was found in \cite{ASS17} that the 
complex Langevin evolution (without the introduction of a kernel and for $z_p$ away from the real axis) converges 
correctly for large enough values of $\beta$, but fails for smaller ones. This failure, despite the absence of boundary terms, was then explained to be due to the 
contribution of unwanted integration cycles \cite{ASS17,SS19}, making the model a
prime target for the new correctness criterion. We emphasize that this is an example for the appearance of unwanted cycles even without using a kernel. Let us now discuss the 
setup in detail.

The density \eqref{eq:one_pole}, in which $n_p$ is assumed to be a positive
integer, $\beta$ is assumed real and $z_p$ is an arbitrary complex number,
has two linearly independent integration cycles, namely
\begin{equation}
    \int_{\gamma_+}dz \equiv \int_{z_p}^{\infty}dz\;, \quad
    \int_{\gamma_-}dz \equiv \int_{-\infty}^{z_p}dz\;,
\end{equation}
where the exact integration paths are again irrelevant due to Cauchy's theorem. 
The only important property is that the paths start at some point in complex 
infinity where $\rho$ vanishes and end in $z_p$, the finite zero of $\rho$.  

The integral over the real axis, which is what we are ultimately after, can be
written as a linear combination of the two. In particular, with the 
definitions
\begin{equation}
    Z_\pm \equiv \int_{\gamma_\pm}dz (z-z_p)^{n_p}e^{-\beta z^2}\;,
\end{equation}
such that $Z=Z_++Z_-$, we can write the desired expectation values as
\begin{equation}
    \langle f\rangle_\rho \equiv \int_{-\infty}^\infty dz f(z)\rho(z) 
    = \frac{Z_+}{Z}\langle f\rangle_+ + \frac{Z_-}{Z}\langle f\rangle_-\;,
\end{equation}
with
\begin{equation}
    \langle f\rangle_\pm \equiv \frac{1}{Z_\pm}\int_{\gamma_\pm}dz f(z)(z-z_p)^{n_p}e^{-\beta z^2}\;.
\end{equation}
Therefore, in the light of the theorem \eqref{eq:sase_theorem}, the complex
Langevin evolution is correct if and only if it produces $a_\pm=\frac{Z_\pm}{Z}$, where we have defined $\gamma_1=\gamma_+$, $\gamma_2=\gamma_-$ and similar for $a_1$ and $a_2$ in \eqref{eq:sase_theorem}.

In the following, we consider $\beta=1.6$, $z_p=\ii$, and $n_p=2$, giving $\frac{Z_\pm}{Z}\approx0.50\pm0.23\ii$. That these values are not reproduced by a complex Langevin simulation (without a kernel), i.e., that the simulation produces incorrect results, was found in \cite{SS19}. In the following, we demonstrate that one may use the new correctness criterion to prove this fact without having to compute the $a_\pm$ explicitly. For the simulation of \eqref{eq:one_pole}, our simulation setup is analogous to the one used in the previous subsection, outlined in \cite{HMS25}, the only difference being that here we restrict the maximum allowed Langevin step size in our adaptive step-size algorithm to $10^{-6}$ instead of $10^{-5}$.

In order to apply the correctness criterion, as before, we first study the validity of the SDEs \eqref{sde1}, which, for the model \eqref{eq:one_pole} and monomial observables $z^k$, read
\begin{equation}\label{eq:dse_one_pole}
    \left\langle Az^k\right\rangle = k\left\langle z^{k-1}\right\rangle + n_p\left\langle\frac{z^k}{z-z_p}\right\rangle - 2\beta\left\langle z^{k+1}\right\rangle = 0\;.
\end{equation}
This equation demonstrates that here
we indeed have to consider a larger space of observables than just the polynomials, since $\cH\neq\cA$ due to the zero of $\rho(z)$ at $z=z_p$. The left-hand side of \eqref{eq:dse_one_pole} for different $k$ is shown in \cref{fig:dse_one_pole} and we find that the SDEs are indeed well satisfied. 

\begin{figure}[t]
    \centering
    \includegraphics[width=0.5\textwidth]{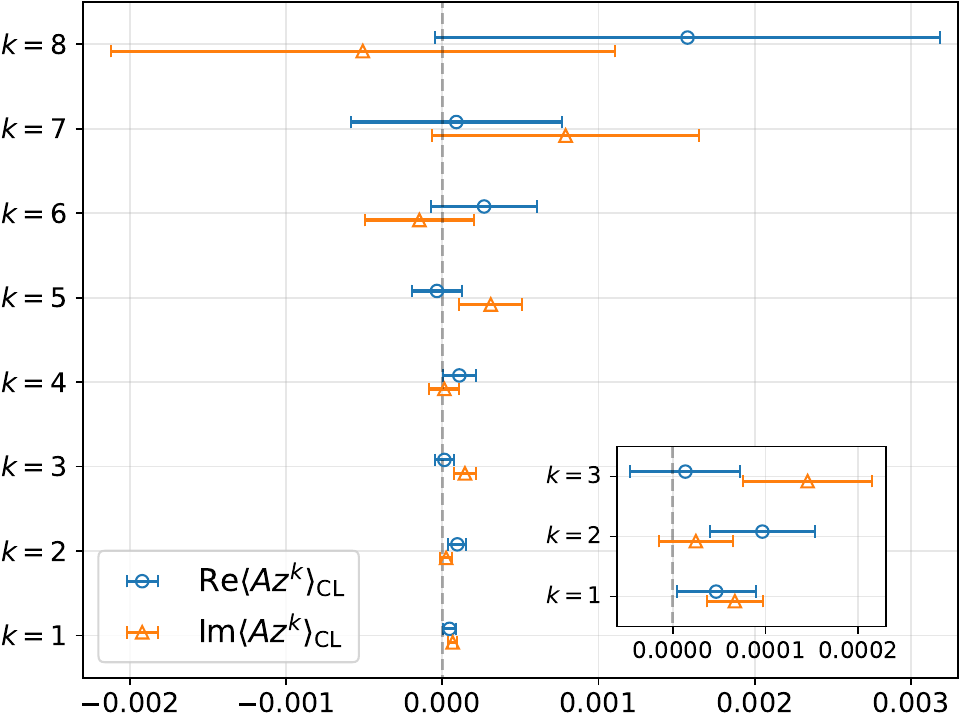}
    \caption{Validity of the SDEs  \eqref{eq:dse_one_pole} in the model \eqref{eq:one_pole} with $\beta=1.6$, $z_p=\ii$, $n_p=2$, $k=1,\dots,8$ and without a kernel. The inset shows a close-up view of the smallest few $k$. The data points have been offset vertically for visual clarity.}
    \label{fig:dse_one_pole}
\end{figure}

Now, in order to show that the complex Langevin results are nonetheless incorrect via the new correctness criterion, it again suffices to find a single
control observable $f$ that violates the bounds in \eqref{bounds}. We define
\begin{equation}
    w(z) = e^{-\frac{\beta z^2}{2}}\;, \quad \rho_r(z) = \frac{1}{Z}(z-z_p)^{n_p}e^{-\frac{\beta z^2}{2}}\;.
\end{equation}
While it turned
out that finding such a polynomial can be a nontrivial task in practice, we were able to obtain at least a handful that serve our purpose. The (notably) simplest one of them reads
\begin{equation}
    f(z) = \sum_{n=0}^{12}c_nz^n\;,
\end{equation}
with the coefficients
\begin{align}
    \begin{aligned}
    c_0 = 210\;, \quad c_1 = -1213\ii\;, \quad c_2 = -3505\;, \quad c_3 = 6331\ii\;, \quad c_4 = 8583\;, \\
    c_5 = -8175\ii\;, \quad c_6 = -6825\;, \quad c_7 = 3813\ii\;, \quad c_8 = 2204\;, \\
    c_9 = -698\ii\;, \quad c_{10} = -300\;, \quad c_{11} = 42.2\ii\;, \quad c_{12} = 14.2\;. 
    \end{aligned}
\end{align}
It was obtained by minimizing the distance\footnote{One possible approach to find suitable polynomials in practice could thus be to start with an arbitrary function for which one can somehow show that it violates the bounds \eqref{bounds} and then approximate it (in the norm \eqref{norms}) with a polynomial of sufficiently high order. We also remark that we are, in fact, not restricted to using polynomials here. Any element of $\cH$ would work in principle.} $\Vert f(z)-g(z)\Vert_w$ to the function $g(z)=e^{-2(1+(2\ii+z)^2)}$, which we found (rather by coincidence) to violate the bounds, within the norm \eqref{norms}. 
For this polynomial, we find that the complex Langevin evolution produces
an expectation value whose absolute value reads $\left\vert\langle f\rangle_{CL}\right\vert=902$ with negligible statistical uncertainties,
while the bound on the right-hand side of \eqref{bounds} reads 
$C\Vert f\Vert_w\approx702$. This bound is thus clearly violated, from 
which one can safely conclude that the complex Langevin results are incorrect.

\section{Towards more realistic models}\label{sec:realistic}
    In this section, we briefly comment on the prospects of applying the new correctness criterion to lattice models. After splitting the
measure of the theory as in \eqref{eq:rho_split}, we need to establish whether the following bound holds in a complex Langevin simulation for all observables $f\in\cH$:
\be
| \langle f \rangle_\textrm{CL} | \le C \Vert f \Vert_w,
\ee
with $\Vert f\Vert_w$ and $C$ defined in \eqref{norms} and 
\eqref{integrable}, respectively.

For a given $f$, determining $\Vert f \Vert_w $ is relatively easy. For instance, one could start at a random point in $ \cM_r$ and find a local maximum
of $ |f(\x,\g) | w(\x) $ using the gradient \emph{ascent} method, since both $ f$ and $ \omega$ are expected to be smooth functions in practice.
To find a global maximum, one starts the gradient ascent
method at $N_\textrm{init}$ initial points and chooses the largest
local maximum as an approximation to the global maximum.
Since  $ |f(\x,\g) | w(\x)  \rightarrow 0 $ as $\vert\x\vert \rightarrow \infty $, this procedure is expected to converge.

In order to compute $C$, let us introduce the non-normalized density $\rho'$ as
\begin{equation}
    \rho(\x,\g) \equiv \frac{1}{Z}\rho'(\x,\g) \quad \textnormal{with} \quad Z \equiv \int_{\cM_r}\rho'(\x,\g)d\mu\;,
\end{equation}
and similar for the reduced density $\rho_r$,
\begin{equation}
    \rho_r(\x,\g)\equiv\frac{1}{Z}\rho_r'(\x,\g)\;,
\end{equation}
such that $\rho'(\x,\g)=w(\x)\rho_r'(\x,\g)$.
Then, we find that 
\begin{equation}
C = \int_{\cM_r} \vert\rho_r(\x,\g)\vert d\mu
 = \frac{\int_{\cM_r} \vert\rho_r'(\x,\g)\vert d\mu}{\left\vert\int_{\cM_r}\rho'(\x,\g)d\mu\right\vert}
\end{equation}
and, thus,
\begin{equation}
C = \frac{\int_{\cM_r} \vert\rho_r'(\x,\g)\vert d\mu}{\left\vert\int_{\cM_r} \frac{\rho'(\x,\g)}{|\rho_r'(\x,\g)|} |\rho_r'(\x,\g)| d\mu\right\vert} = \frac{1}{\left\vert\left\langle \frac{\rho'(\x,\g)}{|\rho_r'(\x,\g)|} \right \rangle _ {|\rho_r'|}\right\vert}\;.
\end{equation}
Here, we have defined the expectation value with respect to $\vert\rho_r'\vert$ as
\begin{equation}
    \langle f(\x,\g)\rangle_{\vert\rho_r'\vert}\equiv \frac{\int_{\cM_r}f(\x,\g)\vert\rho_r'(\x,\g)\vert d\mu}{\int_{\cM_r}\vert\rho_r'(\x,\g)\vert d\mu}\;.
\end{equation}

This means that $C$ can in principle be evaluated in a Monte-Carlo simulation
using the measure $ | \rho_r'(\x,\g)| $ on the manifold $ \cM_r $.
Note, however, that $C$ is the inverse of a quantity that resembles the sign average in a phase-quenched simulation, with
the difference that here we consider $|\rho_r|$ instead of $|\rho |$.
$C$ might thus become large in the case of lattice models with
a strong sign problem.
We mention, however, that one may explicitly introduce the phase $\varphi_r'$ of $\rho_r'=\vert\rho_r'\vert e^{i\varphi_r'}$ to write
\begin{equation}
    \frac{1}{\left\langle\frac{\rho'(\x,\g)}{\vert\rho_r'(\x,\g)\vert}\right\rangle_{\vert\rho_r'\vert}} = \frac{1}{\left\langle e^{i\varphi_r'(\x,\g)}w(\x)\right\rangle_{\vert\rho_r'\vert}}\;.
\end{equation}
The appearance of $w$ as a damping factor might be beneficial in practical computations.
\section{Summary \& conclusions}\label{sec:summary}
    In this work, we have derived a criterion for correctness of
probability measures $dP$ on complex manifolds $\cM_c$, given a complex density $\rho$ on a real manifold $\cM_r$. In particular,
we have found necessary and sufficient conditions that ensure the 
correctness of $dP$, i.e., 
\begin{equation}\label{eq:correctness_summary}
    \int_{\cM_c}dP f(\z,\g_c) = \int_{\cM_r}d\mu\,\rho(\x,\g) f(\x,\g)\;,
\end{equation}
for holomorphic $f$ within a certain algebra of observables $\cH$, defined in \eqref{eq:H}. Here, we assume $\cM_r$ to be of the product form \eqref{eq:M_r}, with $d\mu$ denoting a suitable integration measure, and $\cM_c$ to be its complexification.

More precisely, the necessary and sufficient conditions are on the
one hand the validity of the Schwinger--Dyson equations \eqref{sde1} and \eqref{sde2}, and on the other hand the validity of the bounds \eqref{bounds} for all $f\in\cH$, which need to hold simultaneously. These results were derived using standard theorems of functional analysis, assuming only that $\rho$ is holomorphic and normalizable and that the partial derivatives of $\log\rho$ (with respect to both $\x$ and $\g$) grow at most like a polynomial on $\cM_r$. None of these assumptions is particularly restrictive in the view of usual theories of interest in physics.  Moreover, we have required $\cM_r$ to be free of any zeros of $\rho$, which, however, can always be achieved by a continuous deformation of $\cM_r$. In particular, the presence of a fermionic determinant in the density, as it appears in the path integral treatment of quantum field theories with fermions, does not spoil the proof. Finally, the quantities $w$ and $\rho_r$, defined in \eqref{eq:rho_split}, are to be chosen such that the former obeys \eqref{eq:w_decay} and that the constant $C$ in \eqref{integrable} as well as the norms $\Vert f\Vert_w$ in \eqref{norms} (for all $f\in\cH$) are finite. 

We should mention that there is a lot of freedom 
in the definition of our criteria: On the one hand,  
there is 
freedom in the choices of $w$, $\rho_r$, and $\cM_r$, each different choice giving rise to different bounds in general. Since different manifolds $\cM_r$ that are equivalent via Cauchy's theorem produce different bounds, one may even attempt to take the minimum of \eqref{bounds} over all equivalent $\cM_r$ to find the lowest bound for given $w$, $\rho_r$, and $f$. On the other hand, one may define $\Vert f\Vert_w$ and $C$ using different norms than those in \eqref{norms} and \eqref{integrable}, see, e.g., \eqref{eq:f_norm_general} and \eqref{eq:C_general}.
Any violation of the bounds for any valid choice of $\cM_r$, $w$, $\rho_r$, $f$, as well as of the norms implies that \eqref{eq:correctness_summary} cannot be true.

We have then proceeded to test our chosen criterion in the context of complex Langevin simulations. The complex Langevin expectation values are naturally defined via a probability measure $dP$, obtained from the stationary solution of the real Fokker--Planck equation, hence making them amenable to such an investigation. Indeed, we have found for simple one- and two-dimensional toy models, for which exact results are readily available for comparison, that the criterion is capable of detecting incorrect convergence where conventional correctness criteria, such as boundary-term analyses, fail. In particular, this is the case when the complex Langevin expectation values satisfy the SDEs but receive contributions from unwanted integration cycles $\gamma_i\neq\cM_r$. To the best of our knowledge, there are no other correctness criteria for complex Langevin available that would be sensitive to different integration cycles. Whether unwanted integration cycles play a role in realistic lattice models or even gauge theories remains an open question. A few exploratory results in this direction can be found in \cite{GP09,HMS25}.

Finally, we emphasize once again that our criterion is not suitable for proving that a given simulation produces correct results, as this would require checking an infinite amount of bounds \eqref{bounds}. What we have demonstrated, however, is that it is capable of ruling out correctness for suitable choices of $f$ without the need of computing exact results for observables. How to find such $f$ for a given theory is a nontrivial question, to which we cannot yet provide a definite answer.

\section*{Acknowledgments}
The authors thank Lorenzo Salcedo for useful comments on the manuscript.
This research was funded in whole, or in part, by the Austrian Science Fund (FWF) [10.55776/P36875]. The numerical simulations, from which the results discussed in this work have been computed, were performed on the computing cluster of the University of Graz (GSC). The data analyses are based on 
the \verb|Python| scientific computing ecosystem \cite{python}, in particular via
the packages \cite{numpy,pandas,matplotlib} and we are also grateful for the 
creation and maintenance of the dependencies of the latter. Moreover, auxiliary numerical calculations were 
performed with the \verb|Wolfram Mathematica| computer algebra system \cite{mathematica}.

\section*{Open Access Statement}
For the purpose of open access, the author has applied a CC BY public copyright licence to any Author Accepted Manuscript version arising from this submission.

\section*{Data Availability Statement}
The data set underlying this work is freely available \cite{data} and the analysis scripts
used are published online as well \cite{scripts}. Our simulation code is available upon
request.

\appendix

\section{Proof of continuity of \texorpdfstring{$w$}{w}}\label{app:continuity}
    Here, we prove the continuity of $w$ as defined in \eqref{w_choice}.

\begin{proof}
It suffices to show continuity of $\rho_m(\x)=\sup_{\g\in G^m} |\rho(\x,\g)|$. Because $G^m$ is 
compact and $\rho$ is continuous, the supremum is attained at some $\g_x\in G^m$ for every $\x$, i.e., 
\be
\rho_m(\x)= \vert\rho(\x,\g_x)\vert\;.
\ee
We have to show that for any $\x\in \R^n$ and any $\epsilon>0$ there is a $\delta(\epsilon)>0$ such that
\be
|\rho_m(\x)-\rho_m(\y)|<\epsilon \quad \forall \y\quad {\rm with}\quad |\x-\y|< \delta(\epsilon)\;.
\label{contin}
\ee
According to \cref{lemma:1} below, $|\rho(\x,\g)|$ is continuous in $\x$, uniformly in
$\g\in G^m$, i.e., for any $\epsilon>0$ there is a $\delta(\epsilon)>0$ such that
\be
|\rho(\x,\g_{\x})-\rho(\y,\g_{\x})|<\epsilon \quad \forall \y\quad {\rm with}\quad |\x-\y|<
\delta(\epsilon)\phantom{\;.}
\ee
and
\be
|\rho(\x,\g_{\y})-\rho(\y,\g_{\y})|<\epsilon \quad \forall \y\quad {\rm with}\quad |\x-\y|<
\delta(\epsilon)\;.
\ee
Thus
\bea
\rho_m(\x)&=|\rho(\x,\g_{\x})|<|\rho(\y,\g_{\x})|+\epsilon\le \rho_m(\y)+\epsilon \notag\\
&<|\rho(\x,\g_{\y})|+2\epsilon\le \rho_m(\x)+2\epsilon\;,
\eea
provided $|\x-\y|< \delta(\epsilon)$.
This implies
\be
\rho_m(\x)-\epsilon<\rho_m(\y)<\rho_m(\x)+\epsilon\;,
\ee
which shows the continuity of $w$ defined by \eqref{w_choice}.
\end{proof}

In the proof we used the
\begin{lemma}\label{lemma:1}
$\rho(\x,\g) $ is continuous in $\x$, uniformly in $\g$. This means that for any $\epsilon>0$
there is a $\delta(\epsilon)>0$, independent of $\g$, such that
\be
|\rho(\x,\g)-\rho(\y,\g)|<\epsilon \quad \forall \y\quad {\rm with}\quad |\x-\y|<
\delta(\epsilon)\;.
\ee
\begin{proof}
Since we only need uniformity in $\g$, we use the fact that $\x$ always lies in the interior of a compact set
$B\subset \R^n$. Then the lemma follows from the standard theorem that a continuous function on a compact set is uniformly continuous; see, e.g., \cite{Rud76} (Theorem 4.19). The fact that 
$\delta(\epsilon)$ may depend on $B$ is of no consequence.
\end{proof}
\end{lemma}

\section{Approximation in the \texorpdfstring{$w$}{w} norm}\label{app:closure}
    In this appendix, we demonstrate that the functions $f_\bk$ defined in \eqref{functionform} are contained in the closure $\overline{\cA}$ of $\cA$ with respect to the norm $\Vert\cdot\Vert_w$. 

\begin{proof}
To see this, we have to show that $\exp(\ii\bk\cdot\x)$ is a limit of polynomials, which follows 
by considering the Taylor expansion
\be
\exp(\ii t\bk\cdot\x)=\sum_{n=0}^{N} \frac{(\ii t\bk\cdot\x)^n}{n!}+R_{N}  
\ee
with $t\in\R$ and a remainder $R_{N}$. By using \eqref{eq:w_decay}, one can show that $\Vert R_{N}\Vert_w$ satisfies 
\be
\Vert R_{N}\Vert_w\le \frac{B}{(N+1)!}e^
{-\frac{N+1}{2}} \left(\frac{N+1}{2b}\right)^{\frac{N+1}{2}}|\bk|^{N+1}\;.
\label{taylorbound}
\ee
Here, 
$B$ and $b$ are the constants defined in \eqref{eq:w_decay} and we have
used the general bound on the remainder in Taylor's theorem for a function $f(t)$, given by
\be
|R_{N}|\le M \frac{\left\vert t^{N+1}\right\vert}{(N+1)!}\;,
\ee
provided
\be
\left\vert f^{(N+1)}(\tau)\right\vert\le M\;\quad{\rm for}\;\quad \vert\tau\vert\le\vert t\vert\;. 
\ee
In our case, we have $M=\vert\x\vert^{N+1}\vert\bk\vert^{N+1}$ and $t=1$.
Since $\lim_{N\to\infty}\Vert R_N\Vert_w =0$, this shows that we obtain an arbitrarily good polynomial  approximation in the $\Vert\cdot\Vert_w$ norm for functions of the form \eqref{functionform}. 
\qedhere
\end{proof}

\section{Boundary terms, convergence conditions, and Schwinger--Dyson equations}\label{app:cc_sde}
    Here, we collect a number of facts regarding the relation between vanishing boundary
terms, the validity of the SDEs and the validity of the so-called 
convergence conditions (CCs) \cite{AJS11}. The latter, also known as consistency conditions, are an
expression of the fact that the complex Langevin evolution has reached equilibrium. More precisely, we introduce the following

{\bf Definition:} A linear functional $T_{CC}$ on $\cH$ is said to satisfy the CCs if
\be
(T_{CC},L_{c}f)=0\quad\forall f\in \cH \;,
\label{cc}
\ee
where
\be
L_c\equiv\sum_{i=1}^n A_i\partial_i+\sum_{j=1}^d\sum_{k=1}^m A_{j,k}D_{j,k}
\ee
is the so-called complex Langevin operator (in the second term it was assumed that the group
derivations are with respect to a basis of the Lie algebra which is  orthogonal with respect to
the Killing form) and the $A_i$ and $A_{j,k}$ are the SD operators defined in \eqref{eq:A_i} and \eqref{eq:A_jk}, respectively.

If there are no boundary terms, the CCs \eqref{cc} follow from the statement that the measure $dP$ on $\cM_c$ satisfies the time-independent
Fokker--Planck equation. The converse, however, is not necessarily true, as
there are examples in which the CCs hold for some observables even though they come with nonzero boundary terms. 

If the SDEs are satisfied, the CCs follow automatically, as is trivially shown \cite{AJS11}. Moreover, in \cite{AJS11} it was also argued that the converse is true as well. Unfortunately,
the argument given there is incorrect. Thus, in
the following, we derive a sufficient condition under which the CCs do imply
the SDEs, using the methods employed in the main text. 

{\bf Proposition:} Let $dP$ be a probability measure on $\cM_c$. Assume that $\rho$ does not
vanish on $\cM_r$ and that the derivatives of $\log\rho$ grow at most polynomially in $\vert\x\vert$, that $dP$ obeys the bounds \eqref{bounds}, and that the linear functional on $\cH$ generated by $dP$ satisfies the CCs \eqref{cc}.
Then, for all $f\in\cH$,
\be
\int_{\cM_c}f dP= \int_{\cM_r} f(\x)w(\x)
\tilde\sigma_{CC}(\x,\g) d\mu\;,
\ee
where $\tilde\sigma_{CC}$ is a complex measure  on 
$\cM_r$ of finite total variation. Furthermore, defining 
\be
\sigma_{CC}(\x,\g)\equiv w(\x)\tilde\sigma_{CC}(\x,\g)\;,
\ee
$\sigma_{CC}$ satisfies
\be
L_c^T \sigma_{CC}=0\;,
\label{cfp}
\ee
where $L_c^T$ (the complex Fokker--Planck operator) is the transpose of $L_c$,
\be
L_c^T= -\sum_{i=1}^n \partial_iA_i^T -\sum_{j=1}^d\sum_{k=1}^m D_{j,k}A_{j,k}^T\;.
\ee
\begin{proof}
The proof is a straightforward imitation of the proofs of \cref{prop:1,prop:2} in the main text.
\qedhere \end{proof}

As with the SDEs, \eqref{cfp} is solved by
\be
\sigma_{CC}(\x,\g)= c\rho(\x,\g)\,,\quad c\in\C\,,
\label{sol_cfp}
\ee
where $c=1$ due to the normalization of $P$. Crucially, this solution also obeys the SDEs. If the eigenvalue $0$ of $L_c^T$ is nondegenerate, the SDEs are therefore implied by the CCs. This means that we can formulate the following

{\bf Corollary:} If a bounded linear functional $T_{CC}$ on $\cH$, generated by a probability density $dP$ that satifies the bounds
\eqref{bounds}, obeys the CCs \eqref{cc}, and if the equation \eqref{cfp} has only the
solution \eqref{sol_cfp}, then $T_{CC}$ also obeys the SDEs.

{\bf Remark:} While we cannot rule out that the solution \eqref{sol_cfp} is not unique, this
would mean degeneracy of the zero eigenvalue of $L_c^T$ restricted to the real integration
manifold $\cM_r$. This could not be the case generically; it might at worst happen for exceptional parameter values.
The only cases in which such a degeneracy is known to occur,
are those in which $\rho$ vanishes somewhere on $\cM_r$,
leading to nonergodicity; but we have excluded this possibility.

To summarize, on the one hand, if the complex Langevin process relaxes to an equilibrium distribution for which there are no boundary terms, the CCs hold. 
If in addition the bounds \eqref{bounds} are satisfied and  
the zero eigenvalue of $L_c^T$ is nondegenerate, the SDEs hold. On the other hand, the validity of the SDEs always implies the CCs. For all we know, neither the CCs nor the SDEs can guarantee the absence of boundary terms, nor does the absence of boundary terms automatically imply the SDEs. However, at the time of writing, we have never encountered any counterexamples to either of these scenarios.

\bibliographystyle{JHEP}
\bibliography{bibliography}

\end{document}